\documentclass[a4paper]{article}

\usepackage{amsmath}
\usepackage{amssymb}
\usepackage{amsthm}
\usepackage{alg}
\usepackage[T1]{fontenc}

\newcommand{\MaxkCSP}{\textsc{Max $k$-CSP}}
\newcommand{\MaxkCSPq}{\textsc{Max $k$-CSP$_{q}$}}
\newcommand{\MaxCSP}[1]{\textsc{Max CSP}(#1)}

\newcommand{\R}{\mathbb{R}}
\newcommand{\F}{\mathbb{F}}
\newcommand{\Z}{\mathbb{Z}}
\newcommand{\Proj}{\mathbb{P}}
\newcommand{\scalprod}[1]{\left<  #1 \right> }

\DeclareMathOperator*{\E}{\mathbb{E}}
\DeclareMathOperator*{\Var}{Var}
\DeclareMathOperator*{\Cov}{Cov}
\DeclareMathOperator{\Inf}{Inf}
\DeclareMathOperator{\Support}{Supp}
\DeclareMathOperator{\Opt}{Opt}
\DeclareMathOperator{\Ordo}{\mathcal{O}}

\newcommand{\verifier}{\mathcal{V}}

\theoremstyle{plain}
\newtheorem{theorem}{Theorem}[section]
\newtheorem{lemma}[theorem]{Lemma}
\newtheorem{corollary}[theorem]{Corollary}
\newtheorem{conjecture}[theorem]{Conjecture}
\newtheorem{proposition}[theorem]{Proposition}
\theoremstyle{definition}
\newtheorem{definition}[theorem]{Definition}

\begin{document}


  \title{Approximation Resistant Predicates From Pairwise Independence}
  \date{6 Dec 2007}

  \author{
    Per Austrin\thanks{
      E-mail: \texttt{austrin@kth.se}. Research funded by Swedish Research Council Project Number 50394001.
    }\\
    KTH -- Royal Institute of Technology\\
    Stockholm, Sweden
    \and
    Elchanan Mossel\thanks{
      E-mail: \texttt{mossel@stat.berkeley.edu}. 
Research supported by BSF grant 2004105, NSF CAREER award DMS 0548249 
and DOD ONR grant N0014-07-1-05-06} \\ 
U.C. Berkeley\\
USA}
   
  \maketitle

  \begin{abstract}
    We study the approximability of predicates on $k$ variables from a
    domain $[q]$, and give a new sufficient condition for such
    predicates to be approximation resistant under the Unique Games
    Conjecture.  Specifically, we show that a predicate $P$ is
    approximation resistant if there exists a balanced pairwise
    independent distribution over $[q]^k$ whose support is contained
    in the set of satisfying assignments to $P$.

    Using constructions of pairwise indepenent distributions this result 
    implies that
    \begin{itemize}
    \item
    For general $k \ge 3$ and $q \ge 2$, the
    \MaxkCSPq{} problem is UG-hard to approximate within $q^{\lceil
    \log_2 k +1 \rceil - k} + \epsilon$.  
    \item
    For $k \geq 3$ and $q$ prime power, the hardness ratio is improved 
    to $kq(q-1)/q^k + \epsilon$.  
    \item
    For the special case of $q = 2$, i.e., boolean variables, we can sharpen
    this bound to $(k + \Ordo(k^{0.525}))/2^k + \epsilon$, improving
    upon the best previous bound of $2k/{2^k} + \epsilon$
    (Samorodnitsky and Trevisan, STOC'06) by essentially a factor $2$.
    \item
    Finally, for $q=2$, assuming that the famous Hadamard Conjecture is true, 
    this can
    be improved even further, and the $\Ordo(k^{0.525})$ term can be
    replaced by the constant $4$.
\end{itemize}
  \end{abstract}

  \section{Introduction}

  In the \MaxkCSP{} problem, we are given a set of constraints over a
  set of boolean variables, each constraint being a boolean function
  acting on at most $k$ of the variables.  The objective is to find an
  assignment to the variables satisfying as many of the constraints as
  possible.  This problem is NP-hard for any $k \ge 2$, and as a
  consequence, a lot of research has been focused on studying how well
  the problem can be approximated.  We say that a (randomized)
  algorithm has {\em approximation ratio} $\alpha$ if, for all
  instances, the algorithm is guaranteed to find an assignment which
  (in expectation) satisfies at least $\alpha \cdot \Opt$ of the
  constraints, where $\Opt$ is the maximum number of simultaneously
  satisfied constraints, over any assignment.

  A particularly simple approximation algorithm is the algorithm which
  simply picks a random assignment to the variables.  This algorithm
  has a ratio of $1/2^k$.  It was first improved by Trevisan
  \cite{trevisan98parallel} who gave an algorithm with ratio $2/2^k$
  for \MaxkCSP{}.  Recently, Hast \cite{hast05approximating} gave an
  algorithm with ratio $\Omega(k/(\log k 2^k))$, which was
  subsequently improved by Charikar et
  al.~\cite{charikar06approximation} who gave an algorithm with
  approximation ratio $c \cdot k/2^k$, where $c > 0.44$ is an absolute
  constant.

  The PCP Theorem implies that the \MaxkCSP{} problem is NP-hard to
  approximate within $1/c^k$ for some constant $c > 1$.  Samorodnitsky
  and Trevisan \cite{samorodnitsky00pcp} improved this hardness to
  $2^{2\sqrt{k}}/2^k$, and this was further improved to
  $2^{\sqrt{2k}}/2^k$ by Engebretsen and Holmerin
  \cite{engebretsen05efficient}.  Finally, Samorodnitsky and Trevisan
  \cite{samorodnitsky05gowers} proved that, if the Unique Games
  Conjecture \cite{khot02power} is true, then the 
  \MaxkCSP{} problem is hard to
  approximate within $2k/2^k$.  To be more precise, the hardness they
  obtained was $2^{\lceil \log_2 k+1 \rceil}/2^k$, which is
  $(k+1)/2^k$ for $k = 2^r-1$, but can be as large as $2k/2^k$ for
  general $k$.  Thus, the current gap between hardness and
  approximability is a small constant factor of $2/0.44$.

  For a predicate $P: \{0,1\}^k \rightarrow \{0,1\}$, the $\MaxCSP{P}$
  problem is the special case of \MaxkCSP{} in which all constraints
  are of the form $P(l_1, \ldots, l_k)$, where each literal $l_i$ is
  either a variable or a negated variable.  For this problem, the
  random assignment algorithm achieves a ratio of $m/2^k$, where $m$
  is the number of satisfying assignments of $P$.  Surprisingly, it
  turns out that for certain choices of $P$, this is the best possible
  algorithm.  In a celebrated result, Håstad \cite{håstad01optimal}
  showed that for $P(x_1, x_2, x_3) = x_1 \oplus x_2 \oplus x_3$, the
  \MaxCSP{P} problem is hard to approximate within $1/2+\epsilon$.

  Predicates $P$ for which it is hard to approximate the $\MaxCSP{P}$
  problem better than a random assignment, are called
  \emph{approximation resistant}.  A slightly stronger notion is that
  of \emph{hereditary} approximation resistance -- a predicate $P$ is
  hereditary approximation resistant if all predicates implied by $P$
  are approximation resistant.  A natural and important question is to
  understand the structure of approximation resistance.  For $k = 2$
  and $k = 3$, this question is resolved -- predicates on $2$
  variables are never approximation resistant, and a predicate on $3$
  variables is approximation resistant if and only if it is implied by
  an XOR of the three variables
  \cite{håstad01optimal,zwick98approximation}.  For $k = 4$, Hast
  \cite{hast05beating} managed to classify most of the predicates with
  respect to to approximation resistance, but for this case there does
  not appear to be as nice a characterization as there is in the case
  $k = 3$.  It turns out that, assuming the Unique Games Conjecture,
  most predicates are in fact hereditary approximation resistant -- as
  $k$ grows, the fraction of such predicates tend to $1$
  \cite{håstad06approximation}.  Thus, instead of attempting to
  understand the seemingly complicated structure of approximation
  resistant predicates, one might try to understand the possibly
  easier structure of hereditary approximation resistant predicates,
  as these constitute the vast majority of \emph{all} predicates.

  A natural approach for obtaining strong inapproximability for the
  \MaxkCSP{} problem is to search for approximation resistant
  predicates with very few accepting inputs.  This is indeed how all
  mentioned hardness results for \MaxkCSP{} come about (except the one
  implied by the PCP Theorem).

  It is natural to generalize the \MaxkCSP{} problem to variables over
  a domain of size $q$, rather than just boolean variables.  Without
  loss of generality we may assume that the domain is $[q]$.  We call
  this the \MaxkCSPq{} problem.  For \MaxkCSPq{}, the random
  assignment gives a $1/q^k$-approximation, and any
  $f(k)$-approximation algorithm for the \MaxkCSP{} problem gives a
  $f(k \lceil \log_2 q\rceil)$-approximation algorithm for the
  \MaxkCSPq{} problem.  Thus, Charikar et al.'s algorithm gives a
  $0.44 k \log_2 q/q^k$-approximation in the case that $q$ is a power
  of $2$.  The best previous inapproximability for the \MaxkCSPq{}
  problem is due to Engebretsen
  \cite{engebretsen04nonapproximability}, who showed that the problem
  is NP-hard to approximate within $q^{\Ordo(\sqrt{k})}/q^k$.

  Similarly to $q = 2$, we can define the \MaxCSP{P} problem for $P:
  [q]^k \rightarrow \{0,1\}$.  Here, there are several natural ways of
  generalizing the notion of a literal.  One possible definition is to
  say that a literal $l$ is of the form $\pi(x_i)$, for some variable
  $x_i$ and permutation $\pi: [q] \rightarrow [q]$.  A stricter
  definition is to say that a literal is of the form $x_i + a$, where,
  again, $x_i$ is a variable, and $a \in [q]$ is some constant.  In
  this paper, we use the second, stricter, definition.  As this is a
  special case of the first definition, our hardness results apply
  also to the first definition.


  \subsection{Our contributions}
  Our main result is the following:

  \begin{theorem}
    \label{thm:main}
    Let $P: [q]^k \rightarrow \{0,1\}$ be a $k$-ary predicate over $[q]$, 
    and let $\mu$ be a distribution over $[q]^k$ such that 
    $$
    \Pr_{x \in ([q]^k,\mu)}[P(x)] = 1
    $$
    and for all $1 \leq i \neq j \leq k$ and all $a,b \in [q]$, it holds that 
    $$
    \Pr_{x \in ([q]^k,\mu)}[x_i = a, x_j = b] = 1/q^2.
    $$ 
    Then, for any $\epsilon > 0$, the
    UGC implies that the $\MaxCSP{P}$ problem is NP-hard to
    approximate within
    $$
    \frac{|P^{-1}(1)|}{q^k} + \epsilon,
    $$
    i.e., $P$ is hereditary approximation resistant.
  \end{theorem}

  Using constructions of pairwise independent distributions, we obtain 
  the following corollaries:

  \begin{theorem}
    \label{thm:hardness_general}
    For any $k \ge 3$, $q \ge 2$, and $\epsilon > 0$, it is UG-hard to
    approximate the \MaxkCSPq{} problem within
    $$
    \frac{q^{\lceil \log_2 k+1 \rceil}}{q^k} + \epsilon <
    \frac{k^{\log_2 q} \cdot q}{q^k} + \epsilon.
    $$
    In the special case that $k = 2^r-1$ for some $r$ the hardness
    ratio improves to
    $$\frac{k^{\log_2 q}}{q^k} + \epsilon.$$
  \end{theorem}

  This already constitutes a significant improvement upon the
  $q^{\Ordo(\sqrt{k})}/q^k$-hardness of Engebretsen, and in the case that $q$
  is a prime power we can improve this even further.

  \begin{theorem}
    \label{thm:hardness_primepower}
    For any $k \ge 3$, $q = p^e$ for some prime $p$, and $\epsilon >
    0$, it is UG-hard to approximate the \MaxkCSPq{} problem within
    $$
    \frac{k(q-1)q}{q^k} + \epsilon.
    $$
    In the special case that $k = (q^r-1)/(q-1)$ for some $r$, the hardness
    ratio improves to
    $$\frac{k(q-1)+1}{q^k} + \epsilon \le \frac{kq}{q^k} + \epsilon.$$
  \end{theorem}

  Neither of these two theorems improve upon the results of
  \cite{samorodnitsky05gowers} for the case of $q = 2$.  However, the
  following theorem does.

  \begin{theorem}
    \label{thm:hardness_boolean_variables}
    For any $k \ge 3$ and $\epsilon > 0$, it is UG-hard to approximate
    the \MaxkCSP{} problem within
    $$
    \frac{k + \Ordo(k^{0.525})}{2^k} + \epsilon.
    $$
    If the Hadamard Conjecture is true, it is UG-hard to approximate
    the \MaxkCSP{} problem within
    $$ \frac{4\lceil (k+1)/4 \rceil}{2^k} + \epsilon \le
    \frac{k+4}{2^k} + \epsilon
    $$
  \end{theorem}

  Thus, we improve the hardness of \cite{samorodnitsky05gowers} by
  essentially a factor $2$, decreasing the gap to the best algorithm
  from roughly $2/0.44$ to roughly $1/0.44$.

  \subsection{Related work}
  
  It is interesting to compare our results to the results of
  Samordnitsky and Trevisan~\cite{samorodnitsky05gowers}.  Recall that
  using the Gowers norm, \cite{samorodnitsky05gowers} prove that the
  \MaxkCSP{} problem has a hardness factor of $2^{\lceil \log_2 k+1
  \rceil}/2^k$, which is $(k+1)/2^k$ for $k = 2^r-1$, but can be as
  large as $2k/2^k$ for general $k$.

  Our proof uses the same version of the UGC, but the analysis is 
  more direct and more general. The proof of~\cite{samorodnitsky05gowers}   
  requires us to work specifically with a linearity 
  hyper-graph test for the long codes. For this test, the success probability 
  is shown to be closely related to the Gowers inner product of the long codes. 
  In particular, in the soundness analysis it is shown 
  that if the value of this test is too large, it follows that the Gowers norm 
  is larger than for ``random functions''. 
  From this it is shown that at least two of the functions have large 
  influences which in turns allows us to obtain a good solution for the UGC. 

  Our construction on the other hand allows any pairwise distribution to
  define a long-code test. Using~\cite{mossel07noise}  
  we show that if a collection of supposed long codes 
  does better than random for this long code test, then at least two of them 
  have large influences. 

  Our proof has a number of advantages: first it applies to any
  pairwise independent distribution. This should be compared
  to~\cite{samorodnitsky05gowers} that require us to work specifically
  with the hyper-graph linearity test.  In particular our results
  allow us to obtain hardness results for $\MaxCSP{P}$ for a wide
  range of $P$'s.  The results are general enough to accomodate any
  domain $[q]$ (it is not clear if the results of
  \cite{samorodnitsky05gowers} extend to larger domains), and we are
  also able to obtain a better hardness factor for most values of $k$
  even in the $q = 2$ case.

  Also, our proof uses bounds on expectations of products under
  certain types of correlation, putting it in the same general
  framework as many other UGC-based hardness results, in particular
  those for $2$-CSPs
  \cite{khot07optimal,khot06sdp,austrin07balanced,austrin07towards,odonnell07optimal}.

  Finally, our proof gives parametrized hardness in the following
  sense.  We give a family of hardness assumptions, called the
  $(t,k)$-UGC. All of these assumptions follow from the UGC, and in
  particular the case $t = 2$ is known to be equivalent to the UGC.
  However, the $(t,k)$-UGC assumption is weaker for larger values of
  $t$. For each value of $t$ our results imply a different hardness of
  approximation factor.  Specifically, if the $(t,k)$-UGC is true for
  some $t \ge 3$, then the \MaxkCSP{} problem is NP-hard to
  approximate within $\Ordo\left(k^{\lceil t/2 \rceil -
  1}/2^k\right)$.  Thus, even the $(4,k)$-UGC gives a hardness of
  $\Ordo(k/2^k)$, and for $t < \sqrt{k}/\log k$, the $(t,k)$-UGC gives
  a hardness better than the best unconditional result known
  \cite{engebretsen05efficient}.

  \section{Definitions}


  \subsection{Unique Games}

  We use the following formulation of the Unique Label Cover Problem:
  given is a $k$-uniform hypergraph, where for each edge $(v_1,
  \ldots, v_k)$ there are $k$ permutations $\pi_1, \ldots, \pi_k$ on
  $[L]$.  We say that an edge $(v_1, \ldots, v_k)$ with permutations
  $\pi_1, \ldots, \pi_k$ is $t$-wise satisfied by a labelling $\ell: V
  \rightarrow [L]$ if there are $i_1 < i_2 < \ldots < i_t$ such that
  $\pi_{i_1}(\ell(v_{i_1})) = \pi_{i_2}(\ell(v_{i_2})) = \ldots =
  \pi_{i_t}(\ell(v_{i_t}))$.  We say that an edge is completely
  satisfied by a labelling if it is $k$-wise satisfied.
  
  We denote by $\Opt_t(X) \in [0,1]$ the maximum fraction of $t$-wise
  satisfied edges, over any labelling.  Note that $\Opt_{t+1}(X) \le
  \Opt_{t}(X)$.




  The following conjecture is known to follow 
  from the Unique Games Conjecture (see details below).

  \begin{conjecture}
    For any $2 \le t \le k$, and $\delta > 0$, there exists an $L > 0$
    such that it is NP-hard to distinguish between $k$-ary Unique
    Label Cover instances $X$ with label set $[L]$ with $\Opt_k(X) \ge
    1-\delta$, and $\Opt_t(X) \le \delta$.
  \end{conjecture}

  For particular values of $t$ and $k$ we will refer to the
  corresponding special case of the above conjecture as the
  $(t,k)$-{\em Unique Games Conjecture} (or the $(t,k)$-UGC).

  Khot's original formulation of the Unique Games Conjecture
  \cite{khot02power} is then exactly the $(2,2)$-UGC, and Khot and
  Regev \cite{khot03vertex} proved that this conjecture is equivalent
  to the $(2,k)$-UGC for all $k$, which is what Samorodnitsky and
  Trevisan \cite{samorodnitsky05gowers} used to obtain hardness for
  \MaxkCSP{}.

  In this paper, we mainly use the $(3,k)$-UGC to obtain our hardness
  results.  Clearly, since $\Opt_{t+1}(X) \le \Opt_{t}(X)$, the
  $(t,k)$-UGC implies the $(t+1,k)$-UGC, so our assumption is implied
  by the Unique Games Conjecture.  But whether the converse holds, or
  whether there is hope of proving this conjecture (or, say, the
  $(k,k)$-UGC for large $k$) without proving the Unique Games
  Conjecture, is not clear, and should be an interesting direction for
  future research.

  \subsection{Influences}
  It is well known (see e.g.~\cite{khot07optimal}) 
  that each function $f : [q]^n \to \R$ admits a unique  
  {\em Efron-Stein decomposition:} $f = \sum_{S \subseteq [n]} f_S$ where 
  \begin{itemize}
  \item The function $f_S$ depends on $x_S = (x_i : i \in S)$ only.
  \item For every $S' \not \subseteq S$, and every $y_{S'} \in [q]^{S'}$ it holds that 
    \[
    \E[f_S(x_S) | x_{S'} = y_{S'}] = 0.
    \]
  \end{itemize}
  For $m \leq n$ we write $f^{\leq m} = \sum_{S : |S| \leq m} f_S$ for the 
  $m$-degree expansion of $f$. 
  We now define  the 
  \emph{influence of the $i$th coordinate on $f$},
  denoted by $\Inf_i(f)$ by 
  \begin{equation} \label{eqn:influence}
  \Inf_i(f) = \E_x[\Var_{x_i}[f(x)]]. 
  \end{equation}
  We define the \emph{$m$-degree influence of the $i$th coordinate on $f$},
  denoted by $\Inf_i^{\leq m}(f)$ by $\Inf_i(f^{\leq m})$. 

  Recall that the influence $\Inf_i(f)$ measures how much the
  function $f$ depends on the $i$'th variable, while the low degree
  influences $\Inf_i^{\leq m}(f)$ measures this for the low part of
  the expansion of $f$. The later quantity is closely related to the
  influence of $f$ on ``slightly noisy inputs''.  
  
  An important property of low-degree influences is that
  $$
  \sum_{i=1}^{n} \Inf_{i}^{\le m}(f) \le m \Var[f],
  $$ 
  implying that the number of coordinates with large low-degree
  influence must be small.  In particular, if $f: [q]^n \rightarrow
  [0,1]$, then the the number of coordinates with low-degree influence
  at least $\tau$ is at most $\tau / m$.

  \subsection{Correlated Probability Spaces}

  We will be interested in probability distributions 
  supported in $P^{-1}(1) \subseteq [q]^k$. 
  It would be useful to follow~\cite{mossel07noise} 
  and view $[q]^k$ with such probability measure as a collection of 
  $k$ {\em correlated spaces} corresponding to the $k$ coordinates. 
  We proceed with formal definitions of two and $k$ correlated spaces. 


  \begin{definition}
    Let $(\Omega, \mu)$ be a probability space over a finite product
    space $\Omega = \Omega_1 \times \Omega_2$.  The \emph{correlation}
    between $\Omega_1$ and $\Omega_2$ (with respect to $\mu$) is
    $$
    \rho(\Omega_1, \Omega_2; \mu) = \sup \{ \, \Cov[f_1(x_1)f_2(x_2)] \,:\,f_i: \Omega_i \rightarrow \R, \Var[f_i(x_i)] = 1\,\},
    $$
    where $(x_1,x_2)$ is drawn from $(\Omega, \mu)$.
  \end{definition}

  \begin{definition}
    Let $(\Omega, \mu)$ be a probability space over a finite product
    space $\prod_{i=1}^k \Omega_i$, and let $\Omega_S = \prod_{i \in
    S} \Omega_i$.  The correlation of $\Omega_1, \ldots, \Omega_k$
    (with respect to $\mu$) is
    $$
    \rho(\Omega_1, \ldots, \Omega_k; \mu) = \max_{1 \le i \le k-1}
    \rho(\Omega_{\{1, \ldots, i\}}, \Omega_{\{i+1, \ldots, k\}}; \mu)
    $$
  \end{definition}
  
  Of particular interest to us is the case where correlated spaces 
  are defined by a measure that it $t$-wise independent. 

  \begin{definition}
    Let $(\Omega, \mu)$ be a probability space over a product space
    $\Omega = \prod_{i=1}^k \Omega_i$.  We say that $\mu$ is
    $t$-wise independent if, for any choice of $i_1 < i_2 < \ldots <
    i_t$ and $b_1, \ldots, b_t$ with $b_j \in \Omega_{i_j}$, we have
    that
    $$ \Pr_{w \in (\Omega, \mu)}[w_{i_1} = b_1, \ldots, w_{i_s}=b_s] =
    \prod_{j=1}^t \Pr_{w \in (\Omega, \mu)}[w_{i_j} = b_j]
    $$
    We say that $(\Omega, \mu)$ is \emph{balanced} if for every $i \in
    [k], b \in \Omega_i$, we have that $\Pr_{w \in (\Omega, \mu)}[w_i
    = b] = 1/|\Omega_i|$.
  \end{definition}

The following theorem considers low influence functions that act on 
correlated spaces where the correlation is given by a $t$-wise independent 
probability measure for $t \geq 2$. It shows that in this case, the functions 
have almost the same distribution as if they were completely independent. 
Moreover, the result holds even if some of the functions have large influences 
as long as in each coordinate not more than $t$ functions 
have large influences. 

  \begin{theorem}[\cite{mossel07noise}, Theorem 6.6 and Lemma 6.9]
    \label{thm:correlation_bound}
    Let $(\Omega, \mu)$ be a finite probability space over $\Omega =
    \prod_{i=1}^k \Omega_i$ with the following properties:
    \begin{enumerate}
    \item[(a)] $\mu$ is $t$-wise independent.
    \item[(b)] For all $i \in [k]$ and $b_i \in \Omega_i$, $\mu_i(b_i) > 0$.
    \item[(c)] $\rho(\Omega_1, \ldots, \Omega_k; \mu) < 1$.
    \end{enumerate}
    Then for all $\epsilon > 0$ there exists a $\tau > 0$ and $d > 0$ such that
    the following holds.  Let $f_1, \ldots, f_k$ be functions $f_i: \Omega_i^n \rightarrow
    [0,1]$ satisfying that, for all $1 \le j \le n$,
    $$|\{\,i\,:\,\Inf^{\le d}_{j}(f_i) \ge \tau\,\}| \le t.$$
    Then
    $$\left| \E_{w_1,\ldots,w_n}\left[ \prod_{i=1}^{k} f_i(w_{1,i}, \ldots, w_{n,i}) \right] - \prod_{i=1}^k \E_{w_1,\ldots,w_n}\left[ f_i(w_{1,i}, \ldots, w_{n,i}) \right]
    \right| \le \epsilon,$$
    where $w_1, \ldots, w_n$ are drawn independently from $(\Omega,
    \mu)$, and $w_{i,j} \in \Omega_{j}$ denotes the $j$th coordinate of
    $w_i$.
  \end{theorem}
  
  Note that a sufficient condition for (c) to hold in the above
  theorem is that for all $w \in \Omega$, $\mu(w) > 0$.

  Roughly speaking, the basic idea behind the theorem and its proof is
  that low influence functions cannot detect dependencies of high
  order -- in particular if the underlying measure is pairwise
  independent, then low influence functions of different coordinates
  are essentially independent.

  \section{Main theorem}

  In this section, we prove our main theorem. Note that it is a 
  generalization of Theorem~\ref{thm:main}.

  \begin{theorem}
    \label{thm:main_general}
    Let $P: [q]^k \rightarrow \{0,1\}$ be a $k$-ary predicate over a
    (finite) domain of size $q$, and let $\mu$ be a balanced $t$-wise
    independent distribution over $[q]^k$ such that $\Pr_{x \in
    ([q]^k,\mu)}[P(x)] > 0$.  Then, for any $\epsilon > 0$, the
    $(t+1,k)$-UGC implies that the $\MaxCSP{P}$ problem is NP-hard to
    approximate within
    $$
    \frac{|P^{-1}(1)|}{q^k \cdot \Pr_{x \in ([q]^k,\mu)}[P(x)]} + \epsilon
    $$
  \end{theorem}
  
  In particular, note that if $\Pr_{x \in ([q]^k,\mu)}[P(x)] = 1$, i.e., if
  the support of $\mu$ is entirely contained in the set of
  satisfying assignments to $P$, then $P$ is approximation resistant.
  It is also hereditary approximation resistant, since the support of
  $\mu$ will still be contained in $P^{-1}(1)$ when we add more
  satisfying assignments to $P$.
  
\paragraph{Reduction.}
  Given a $k$-ary Unique Label Cover instance $X$, the prover writes
  down the table of a function $f_v: [q]^L \rightarrow [q]$ for each
  $v$, which is supposed to be the long code of the label of the
  vertex $v$.  Furthermore, we will assume that $f_v$ is folded, i.e.,
  that for every $x \in [q]^k$ and $a \in [q]$, we have $f_v(x +
  (a,\ldots,a)) = f_v(x) + a$ (where the definition of ``$+$'' in
  $[q]$ is arbitrary as long as $([q],+)$ is an Abelian group).  When
  reading the value of $f_v(x_1, \ldots, x_L)$, the verifier can
  enforce this condition by instead querying $f_v(x_1 - x_1, x_2-x_1,
  \ldots, x_L-x_1)$ and adding $x_1$ to the result.  Let $\eta > 0$ be
  a parameter, the value of which will be determined later, and define
  a probability distribution $\mu'$ on $[q]^k$ by 
  $$\mu'(w) = (1-\eta)\cdot \mu(w) + \eta \cdot \mu_{U}(w),$$ where
  $\mu_U$ is the uniform distribution on $[q]^k$, i.e., $\mu_U(w) =
  1/q^k$.  Given a proof $\Sigma = \{f_v\}_{v \in V}$ of supposed long
  codes for a good labelling of $X$, the verifier checks $\Sigma$ as
  follows.
 
  \begin{algorithm}[!h]
    \caption{The verifier $\verifier$}
    \algname{$\verifier$}{$X$, $\Sigma = \{f_v\}_{v \in V}$}
    
    \begin{algtab}[0.7cm]
      Pick a random edge $e = (v_1, \ldots, v_k)$ with permutations
      $\pi_1, \ldots, \pi_k$.\\

      For each $i \in [L]$, draw $w_i$ randomly from $([q]^k, \mu')$.\\

      For each $j \in [k]$, let $x_j = w_{1,j} \ldots w_{L,j}$, and
      let $b_j = f_{v_j}\pi_j(x_j)$.\\

      Accept if $P(b_1, \ldots, b_k)$.
    \end{algtab}
  \end{algorithm}

  \begin{lemma}[Completeness]
    \label{lemma:completeness}
    For any $\delta$, if $\Opt_{k}(X) \ge 1-\delta$, then there is a
    proof $\Sigma$ such that
    $$\Pr[\textrm{$\verifier(X, \Sigma)$ accepts}] \ge
    (1-\delta)(1-\eta)\Pr_{w \in ([q]^k,\mu)}[P(w)]
    $$
  \end{lemma}

  \begin{proof}
    Take a labelling $\ell$ for $X$ such that a fraction $\ge
    1-\delta$ of the edges are $k$-wise satisfied, and let $f_v: [q]^L
    \rightarrow [q]$ be the long code of the label $\ell(v)$ of vertex $v$.

    Let $(v_1, \ldots, v_k)$ be an edge that is $k$-wise satisfied by
    $\ell$.  Then $f_{v_1} \pi_1 = f_{v_2} \pi_2 = \ldots = f_{v_k}
    \pi_k$, each being the long code of $i := \pi_1(\ell(v_1))$.  The
    probability that $\verifier$ accepts is then exactly the probability
    that $P(w_{i})$ is true, which, since $w_i$ is drawn from $([q]^k,
    \mu)$ with probability $1-\eta$, is at least $(1-\eta)\Pr_{w \in
    ([q]^k,\mu)}[P(w)]$.

    The probability that the edge $e$ chosen by the verifier in step
    $1$ is satisfied by $\ell$ is at least $1-\delta$, and so we end
    up with the desired inequality.
  \end{proof}

  \begin{lemma}[Soundness]
    \label{lemma:soundness}
    For any $\epsilon > 0$, $\eta > 0$, there is a constant $\delta :=
    \delta(\epsilon, \eta, t, k, q) > 0$, such that if $\Opt_{t+1}(X) <
    \delta$, then for any proof $\Sigma$, we have
    $$\Pr[\textrm{$\verifier(X, \Sigma)$ accepts}] \le
    \frac{|P^{-1}(1)|}{q^k} + \epsilon$$
  \end{lemma}

  \begin{proof}
    Assume that 
    \begin{equation}
      \label{eqn:soundness_assumption}
      \Pr[\textrm{$\verifier(X, \Sigma)$ accepts}] >
      \frac{|P^{-1}(1)|}{q^k} + \epsilon.
    \end{equation}
    We need to prove that this implies that there is a $\delta :=
    \delta(\epsilon, \eta, t, k, q) > 0$ such that $\Opt_{t+1}(X)
    \ge \delta$.

    Equation~\ref{eqn:soundness_assumption} implies that for a
    fraction of at least $\epsilon/2$ of the edges $e$, the
    probability that $\verifier(X, \Sigma)$ accepts when choosing $e$
    is at least $\frac{|P^{-1}(1)|}{q^k} + \epsilon/2$.

    Let $e = (v_1, \ldots, v_k)$ with permutations $\pi_1, \ldots,
    \pi_k$ be such a ``good'' edge.  For $v \in V$ and $a \in [q]$,
    define $g_{v,a}: [q]^L \rightarrow \{0,1\}$ by
    $$
    g_{v,a}(x) = \left\{
    \begin{array}{ll}
      1 & \textrm{if $f_{v}(x) = a$}\\
      0 & \textrm{otherwise}
    \end{array}\right..
    $$
    The probability that $\verifier$ accepts when choosing $e$ is then
    exactly
    $$
    \sum_{x \in P^{-1}(1)} \E_{w_1,\ldots,w_{L}}\left[ \prod_{i=1}^k g_{v_i,x_i}\pi_i(w_{1,i}, \ldots, w_{L,i}) \right],
    $$ 
    which, by the choice of $e$, is greater than $|P^{-1}(1)|/q^k +
    \epsilon/2$.  This implies that there is some $x \in P^{-1}(1)$
    such that
    \begin{eqnarray*}
      \E_{w_1,\ldots,w_{L}}\left[ \prod_{i=1}^k g_{v_i,x_i} \pi_i(w_{1,i}, \ldots, w_{L,i}) \right] &>& 1/q^k + \epsilon' \\
      &=& \prod_{i=1}^k \E_{w_1, \ldots, w_L} [ g_{v_i,x_i} \pi_i (w_{1,i}, \ldots, w_{L,i}) ] + \epsilon',
    \end{eqnarray*}
    where $\epsilon' = \epsilon / 2 / |P^{-1}(1)|$ and the last
    equality uses that, because $f_{v_i}$ is folded and $\mu$ is
    balanced, we have $\E_{w_1, \ldots, w_L} [
    g_{v_i,x_i}(w_{1,i}, \ldots, w_{L,i}) ] = 1/q$.

    Note that because both $\mu$ and $\mu_U$ are $t$-wise independent,
    $\mu'$ is also $t$-wise independent.  Also, we have that for each
    $w \in [q]^k$, $\mu'(w) \ge \eta / q^k > 0$, which
    implies both conditions (b) and (c) of
    Theorem~\ref{thm:correlation_bound}.  Then, the contrapositive
    formulation of Theorem~\ref{thm:correlation_bound} implies that
    there is an $i \in [L]$ and at least $t+1$ indices $J \subseteq
    [k]$ such that $\Inf^{\le d}_{\pi_j^{-1}(i)}(g_{v_j,x_j}) =
    \Inf^{\le d}_{i}(g_{v_j,x_{j}} \pi_j) \ge \tau$ for all $j \in J$,
    where $\tau$ and $d$ are functions of $\epsilon$, $\eta$, $t$,
    $k$, and $q$.

    The process of constructing a good labelling of $X$ from this
    point is standard.  For completeness, we give a proof in the
    appendix.  Specifically,
    Lemma~\ref{lemma:labelling_from_influences} gives that
    $\Opt_{t+1}(X) \ge \epsilon/2 \left(\frac{\tau}{d \cdot
    q}\right)^{t+1}$, which is a function of $\epsilon$, $\eta$,
    $t$, $k$, and $q$, as desired.
  \end{proof}

  It is now straightforward to prove Theorem~\ref{thm:main_general}.
  
  \begin{proof}[Proof of Theorem~\ref{thm:main_general}]
    Let $c = \Pr_{x \in ([q]^k, \mu)}[P(x)]$, $s = |P^{-1}(1)|/q^k$
    and $\eta = \min(1/4, \frac{\epsilon c}{4 s})$.  Note that since
    the statement of the Theorem requires $c > 0$ we also have $s > 0$
    and $\eta > 0$.  Assume that the $(t+1,k)$-UGC is true, and pick
    $L$ large enough so that it is NP-hard to distinguish between
    $k$-ary Unique Label Cover instances $X$ with $\Opt_{t+1}(X) \le
    \delta$ and $\Opt_k(X) \ge 1-\delta$, where $\delta = \min(\eta,
    \delta(\epsilon c/4, \eta, t, k, q))$, where $\delta( \ldots )$
    is the function from Lemma~\ref{lemma:soundness}.  By Lemmas
    \ref{lemma:completeness} and \ref{lemma:soundness}, we then get
    that it is NP-hard to distinguish between $\MaxCSP{P}$ instances
    with $\Opt \ge (1-\delta)(1-\eta) c \ge (1-2\eta) c$ and $\Opt \le
    s + \epsilon c / 4$.  In other words, it is NP-hard to approximate
    the $\MaxCSP{P}$ problem within a factor
    $$ \frac{s + \epsilon c/4}{(1-2\eta) c} \le \frac{s(1+4\eta)}{c} +
    (1 + 4\eta)\epsilon/4 \le s/c + \epsilon
    $$
  \end{proof}

  \section{Inapproximability for \MaxkCSPq}

  As a simple corollary to Theorem~\ref{thm:main_general}, we have:

  \begin{corollary}
    \label{corollary:pairwise_hardness}
    Let $t \ge 2$ and let $\mu$ be a balanced $t$-wise independent
    distribution over $[q]^k$.
    Then the $(t+1,k)$-UGC implies that that \MaxkCSPq{} problem is
    NP-hard to approximate within
    $$
    \frac{|\Support(\mu)|}{q^k}
    $$
  \end{corollary}

  Thus, we have reduced the problem of obtaining strong
  inapproximability for \MaxkCSPq{} to the problem of finding small
  $t$-wise independent distributions.  As we are mainly interested in
  the strongest possible results that can be obtained by this method,
  our main focus will be on pairwise independence, i.e, $t = 2$.
  However, let us first mention two simple corollaries for general
  values of $t$.  

  For $q = 2$, it is well-known that the binary BCH code gives a
  $t$-wise independent distribution over $\{0,1\}^k$ with support size
  $\Ordo(k^{\lfloor t/2 \rfloor})$ \cite{alon86fast}.  In other words,
  the $(t+1,k)$-UGC implies that the \MaxkCSP{} problem is NP-hard to
  approximate within $\Ordo(k^{\lceil t/2 \rceil}/2^k)$.  Note in
  particular that the $(4,k)$-UGC suffices to get a hardness of
  $\Ordo(k/2^k)$ for \MaxkCSP{}, which is tight up to a constant
  factor.  

  For $q$ a prime power and large enough so that $q \ge k$, there are
  $t$-wise independent distributions over $[q]^k$ with support size
  $q^t$ based on evaluating a random degree-$t$ polynomial over
  $\F_q$.  Thus, in this setting, the $(t+1,k)$-UGC implies a hardness
  factor of $q^{t-k}$ for the \MaxkCSPq{} problem.

  In the remainder of this section, we will focus on the details of
  constructions of pairwise independence, giving hardness for
  \MaxkCSPq{} under the $(3,k)$-UGC.

  \subsection{Theorems \ref{thm:hardness_general} and \ref{thm:hardness_primepower}}

  The pairwise independent distributions used to give Theorems
  \ref{thm:hardness_general} and \ref{thm:hardness_primepower} are
  both based on the following simple lemma, which is well-known but
  stated here in a slightly more general form than usual:

  \begin{lemma}
    \label{lemma:independence}
    Let $R$ be a finite commutative ring, and let $u, v \in R^n$ be
    two vectors over $R$ such that $u_i v_j - u_j v_i \in R^*$ for
    some $i, j$.\footnote{$R^*$ denotes the set of units of $R$.  In
    the case that $R$ is a field, the condition is equivalent to
    saying that $u$ and $v$ are linearly independent.}  Let $X \in
    R^n$ be a uniformly random vector over $R^n$ and let $\mu$ be the
    probability distribution over $R^2$ of $(\scalprod{u,X},
    \scalprod{v,X}) \in R^2$.  Then $\mu$ is a balanced pairwise
    independent distribution.
  \end{lemma}

  \begin{proof}
    Without loss of generality, assume that $i = 1$ and $j = 2$.  It
    suffices to prove that, for all $(a,b) \in R^2$ and any choice of
    values of $X_3, \ldots, X_n$, we have
    $$\Pr[(\scalprod{u,X},\scalprod{v,X}) = (a,b) \,|\,X_3, \ldots,
      X_n] = 1/|R|^2.$$
    For this to be true, we need that the system
    $$
    \left\{
    \begin{array}{l c l c l}
      u_1X_1 &+& u_2X_2 &=& a'\\
      v_1X_1 &+& v_2X_2 &=& b'\\
    \end{array}
    \right.
    $$    
    has exactly one solution, where $a' = a - \sum_{i=3}^n u_iX_i$ and
    similarly for $b'$.  This in turn follows directly from the
    condition on $u$ and $v$.
  \end{proof}

  Consequently, given a set of $m$ vectors in $R^n$ such that any pair
  of them satisfy the condition of Lemma~\ref{lemma:independence}, we
  can construct a pairwise independent distribution over $R^m$ with
  support size $|R|^n$.

  Let us now prove Theorem~\ref{thm:hardness_general}.

  \begin{proof}[Proof of Theorem~\ref{thm:hardness_general}]
    Let $r = \lceil \log_2 k+1 \rceil$.  For a nonempty $S \subseteq
    [r]$, let $u_S \in \Z_q^r$ be the characteristic vector of $S$,
    i.e., $u_{S,i} = 1$ if $i \in S$, and $0$ otherwise.  Then, for
    any $S \ne T$, the vectors $u_S$ and $u_T$ satisfy the condition
    of Lemma~\ref{lemma:independence}, and thus, we have that
    $(\scalprod{u_S, X})_{S \subseteq [r]}$ for a uniformly random $X
    \in \Z_q^r$ induces a balanced pairwise independent distribution
    over $\Z_q^{2^r-1}$, with support size $q^r$.

    When $k=2^r-1$ we get a hardness of $q^{\log_2(k)-k}$, but for
    general values of $k$, in particular $k=2^{r-1}$, we may lose up
    to a factor $q$.
  \end{proof}

  We remark that for $q=2$ this construction gives exactly the
  predicate used by Samorodnitsky and Trevisan
  \cite{samorodnitsky05gowers}, giving an inapproximability of
  $2k/2^k$ for all $k$, and $(k+1)/2^k$ for all $k$ of the form
  $2^l-1$.


  


  Intuitively, it should be clear that when we have more structure on
  $R$ in Lemma~\ref{lemma:independence}, we should be able to find a
  larger collection of vectors where every pair satisfies the
  ``independence condition''.  This intuition leads us to
  Theorem~\ref{thm:hardness_primepower}, dealing with the special case
  of Theorem~\ref{thm:hardness_general} in which $q$ is a prime power.
  The construction of Theorem~\ref{thm:hardness_primepower} is
  essentially the same as that of \cite{obrien80pairwise}.

  \begin{proof}[Proof of Theorem~\ref{thm:hardness_primepower}]
    Let $r = \lceil \log_q(k(q-1)+1) \rceil$, and $n = (q^r-1)/(q-1)
    \ge k$.
    
    Let $\Proj(\F_q^r)$ be the projective space over $\F_q^r$, i.e.,
    $\Proj(\F_q^r) = (\F_q^r \setminus 0) / {\sim}$.  Here $\sim$ is the
    equivalence relation defined by $(x_1, \ldots, x_r) \sim (y_1,
    \ldots, y_r)$ if there exists a $c \in \F_q^*$ such that $x_i = c
    y_i$ for all $i$, i.e., if $(x_1, \ldots, x_r)$ and $(y_1, \ldots,
    y_r)$ are linearly independent.  We then have $|\Proj(\F_q^r)| =
    (q^r-1)/(q-1) = n$.  

    Choose $n$ vectors $u_1, \ldots, u_n \in \F_q^r$ as representatives
    from each of the equivalence classes of $\Proj(\F_q^r)$.  Then any
    pair $u_i, u_j$ satisfy the condition of
    Lemma~\ref{lemma:independence}, and as in
    Theorem~\ref{thm:hardness_general}, we get a balanced pairwise
    independent distribution over $\F_q^{n}$, with support size $q^r$.

    When $k = (q^r-1)/(q-1)$, this gives a hardness of $k(q-1)+1$, and
    for general $k$, in particular $k = (q^{r-1}-1)/(q-1)+1$, we lose
    a factor $q$ in the hardness ratio.
  \end{proof}

  Again, for $q = 2$, this construction gives the same predicate used
  by Samorodnitsky and Trevisan.  In the case that $q \ge k$, we get a
  hardness of $q^2/q^k$, the same factor as we get from the general
  construction for $t$-wise independence mentioned at the beginning of
  this section.






  \subsection{Theorem~\ref{thm:hardness_boolean_variables}}

  Let us now look closer at the special case of boolean variables,
  i.e., $q = 2$.  So far, we have only given a different proof of
  Samorodnitsky and Trevisan's result, but we will now show how to
  improve this.

  An Hadamard matrix is an $n \times n$ matrix over $\pm 1$ such that
  $H H^T = nI$, i.e., each pair of rows, and each pair of columns, are
  orthogonal.  Let $h(n)$ denote the smallest $n' \ge n$ such that
  there exists an $n' \times n'$ Hadamard matrix.  It is a well-known
  fact that Hadamard matrices give small pairwise independent
  distributions and thus give hardness of approximating \MaxkCSP{}.
  To be specific, we have the following proposition:
  \begin{proposition}
    \label{prop:hadamard_gives_hardness}
    For every $k \ge 3$, the $(3,k)$-UGC implies that the
    \MaxkCSP{} problem is UG-hard to approximate within $h(k+1)/2^k +
    \epsilon$.
  \end{proposition}

  \begin{proof}
    Let $n = h(k+1)$ and let $A$ be an $n \times n$ Hadamard matrix,
    normalized so that one column contains only ones.  Remove $n-k$ of
    the columns, including the all-ones column, and let $A'$ be the
    resulting $n \times k$ matrix.  Let $\mu: \{-1,1\}^k \rightarrow
    [0,1]$ be the probability distribution which assigns probability
    $1/n$ to each row of $A'$.  Then $\mu$ is a balanced pairwise
    independent distribution with $|\Support(\mu)| = h(k+1)$.
  \end{proof}

  It is well known that Hadamard matrices can only exist for $n = 1$,
  $n = 2$, and $n \equiv 0 \pmod 4$.  The famous \emph{Hadamard
  Conjecture} asserts that Hadamard matrices exist for all $n$ which
  are divisible by $4$, in other words, that $h(n) = 4\lceil n/4
  \rceil \le n+3$.  It is also possible to get useful unconditional
  bounds on $h(n)$.  We now give one such easy bound.

  \begin{theorem}[\cite{payley33orthogonal}]
    \label{thm:payley}
    For every odd prime $p$ and integers $e,f \ge 0$, there exists an
    $n \times n$ Hadamard matrix $H_n$ where $n = 2^e(p^f+1)$, whenever
    this number is divisible by $4$.
  \end{theorem}

  \begin{theorem}[\cite{baker01difference}]
    \label{thm:primegap}
    There exists an integer $n_0$ such that for every $n \ge n_0$,
    there is a prime $p$ between $n$ and $n + n^{0.525}$.
  \end{theorem}

  \begin{corollary}
    \label{corollary:hadamard_bound}
    We have:
    $h(n) \le n + \Ordo(n^{0.525})$.
  \end{corollary}

  \begin{proof}
    Let $p$ be the smallest prime larger than $n/2$, and let $n' =
    2(p+1) \ge n$.  Then, Theorem~\ref{thm:payley} asserts that there
    exists an $n' \times n'$ Hadamard matrix, so $h(n) \le n'$.  If $n$
    is sufficiently large ($n \ge 2n_0$), then by
    Theorem~\ref{thm:primegap}, $p \le n/2 + (n/2)^{0.525}$ and $n' \le
    n + 2n^{0.525}$, as desired.
  \end{proof}

  Theorem~\ref{thm:hardness_boolean_variables} follows from
  Proposition~\ref{prop:hadamard_gives_hardness} and
  Corollary~\ref{corollary:hadamard_bound}.  

  It is probably possible to get a stronger unconditional bound on
  $h(n$) than the one given by
  Corollary~\ref{corollary:hadamard_bound}, by using stronger
  construction techniques than the one of Theorem~\ref{thm:payley}.



  \section{Discussion}

  We have given a strong sufficient condition for predicates to be
  hereditary approximation resistant under (a weakened version of) the
  Unique Games Conjecture: it suffices for the set of satisfying
  assignments to contain a balanced pairwise independent distribution.
  Using constructions of small such distributions, we were then able
  to construct approximation resistant predicates with few accepting
  inputs, which in turn gave improved hardness for the \MaxkCSPq{}
  problem.

  There are several aspects here where there is room for interesting
  further work:

  As mentioned earlier, we do not know whether the $(t,k)$-UGC implies
  the ``standard'' UGC for large values of $t$.  In particular,
  proving the $(t,k)$-UGC for some $t < \sqrt{k}/\log k$ would give
  hardness for \MaxkCSP{} better than the best current NP-hardness
  result.  But even understanding the $(k,k)$-UGC seems like an
  interesting question.

  A very natural and interesting question is whether our condition is
  also necessary for a predicate to be hereditary approximation
  resistant, i.e., if pairwise independence gives a complete
  characterization of hereditary approximation resistance.

  Finally, it is natural to ask whether our results for \MaxkCSPq{}
  can be pushed a bit further, or whether they are tight.  For the
  case of boolean variables, Hast \cite{hast05beating} proved that any
  predicate accepting at most $2\lfloor k/2 \rfloor + 1$ inputs is
  \emph{not} approximation resistant.  For $k \equiv 2,3 \pmod 4$
  this exactly matches the result we get under the UGC and the
  Hadamard Conjecture (which for $k = 2^r-1$ and $k = 2^r-2$ is the
  same hardness as \cite{samorodnitsky05gowers}).  For $k \equiv 0,1
  \pmod 4$, we get a gap of $2$ between how few satisfying
  assignments an approximation resistant predicate can and cannot
  have.

  Thus, the hitherto very succesful approach of obtaining hardness for
  \MaxkCSP{} by finding ``small'' approximation resistant predicate,
  can not be taken further, but there is still a small constant gap of
  roughly $1/0.44$ to the best current algorithm.  It would be
  interesting to know whether the algorithm can be improved, or
  whether the hardest instances of \MaxkCSP{} are not $\MaxCSP{P}$
  instances for some approximation resistant $P$.

  For larger $q$, this situation becomes a lot worse.  When $q=2^l$ and
  $k=(q^r-1)/(q-1)$, we have a gap of $\Theta(q/\log_2 q)$ between the
  best algorithm and the best inapproximability, and for general
  values of $q$ and $k$, the gap is even larger.

  \bibliographystyle{plain}
  \bibliography{references,austrin}

\begin{appendix}
  \section{Good labellings from influential variables}

  \begin{lemma}
    \label{lemma:labelling_from_influences}
    Let $X$ be a $k$-ary Unique Label Cover instance.  Furthermore,
    for each vertex $v$, let $f_v: [q]^k \rightarrow [q]$ and define
    $$
    g_{v,a}(x) = \left\{
    \begin{array}{ll}
      1 & \textrm{if $f_{v_i} = a$}\\
      0 & \textrm{otherwise}
    \end{array}\right..
    $$
    Then if there is a fraction of at least $\epsilon$ edges $e =
    (v_1, \ldots, v_k)$ with a vector $a \in [q]^k$, an index $i \in
    [L]$ and a set $J \subseteq [k]$ of $|J| = t$ indices such that
    \begin{equation}
      \label{eqn:influence_condition}
      \Inf^{\le d}_{\pi_{j}^{-1}(i)}(g_{v_j,a_j}) \ge \tau
    \end{equation}
    for all $j \in J$, then $\Opt_{t}(X) \ge \delta := \epsilon
    \left(\frac{\tau}{d \cdot q}\right)^t$.
  \end{lemma}
  
  \begin{proof}
    For each $v \in V$, let
    $$C(v) = \{\,i\,|\,\Inf^{\le
      d}_{i}(g_{v,a}) \ge \tau \textrm{ for some $a \in [q]$}\,\}.
    $$
    Note that $|C(v)| \le q \cdot d / \tau$.

    Define a labelling $\ell: V \rightarrow [L]$ by picking, for each
    $v \in V$, a label $\ell(v)$ uniformly at random from $C(v)$ (or
    an arbitrary label in case $C(v)$ is empty).  Let $e = (v_1,
    \ldots, v_k)$ be an edge satisfying
    Equation~\ref{eqn:influence_condition}.  Then for all $j \in J$,
    $\pi_{j}^{-1}(i) \in C(v_j)$, and thus, the probability that
    $\pi_j(\ell(v_j)) = i$ is $1 / |C(v_j)|$.  This implies that the
    probability that this edge is $t$-wise satisfied is at least
    $\prod_{j \in J} 1/|C(v_j)| \ge \left(\frac{\tau}{d \cdot
    q}\right)^t$.  Overall, the total expected number of edges that
    are $t$-wise satisfied by $\ell$ is at least $\delta = \epsilon
    \left(\frac{\tau}{d \cdot q}\right)^t$, and thus $\Opt_{t}(X)
    \ge \delta$.
  \end{proof}
    
\end{appendix}

\end{document}